\pdfoutput=1
\documentclass[
	fontsize=12pt,
	oneside,
	DIV=classic,
	paper=a4,
	pagesize=auto,
	titlepage=yes,
	bibliography=totoc,
	listof=totocnumbered,
	parskip=half,
	abstract=yes
	]{scrartcl}

\newcommand{\cpfversion}{Version 0.4.2}
\newcommand{\cpfbuild}{20260512}

\usepackage[utf8]{inputenc}
\usepackage[T1]{fontenc}
\usepackage[english]{babel}
\usepackage{times}
\usepackage[activate={true,nocompatibility},final,tracking=false,kerning=true,spacing=false,factor=1100,stretch=10,shrink=10]{microtype}

\usepackage[intlimits]{amsmath}
\usepackage{amssymb}
\usepackage{amsthm}
\usepackage{enumerate}
\usepackage{booktabs}
\usepackage{csquotes}

\usepackage[headsepline]{scrlayer-scrpage}
\clearpairofpagestyles
\ohead{Fries, Christian~P.}
\ihead{Faster Forward Sensitivities}
\ifoot{\tiny \copyright 2026 Christian Fries}
\ofoot{\tiny \cpfversion\ (\cpfbuild)}

\usepackage{hyperref}
\hypersetup{
	pdfpagelayout=TwoPageRight,
	pdfborderstyle={/S/U/W 1},
	colorlinks = true,
	linktocpage = true,
	linkcolor=blue,
	citecolor=blue,
	urlcolor=blue,
	pdfauthor={Christian P. Fries},
	pdfsubject={Reduced stochastic hedge ratios from pathwise algorithmic differentiation},
	pdftitle={Faster Forward Sensitivities: Reduced stochastic hedge ratios from pathwise algorithmic differentiation},
	pdfstartview=FitH
}

\usepackage{orcidlink}

\newtheorem{proposition}{Proposition}[section]
\newtheorem{remark}{Remark}[section]

\newcommand{\Rset}{\mathbb{R}}

\newcommand{\norm}[1]{{\lVert #1 \rVert}}
\newcommand{\indicatorfcn}{\mathbf{1}}
\newcommand{\transposed}{\top}
\newcommand{\la}{\langle}
\newcommand{\ra}{\rangle}
\newcommand{\dd}{\mathrm{d}}
\newcommand{\argmin}{\operatorname*{arg\,min}}
\newcommand{\spanop}{\operatorname{span}}

\newcommand{\Piop}{\Pi}

\begin{document}

\author{
	Christian P. Fries%
	\textsuperscript{\orcidlink{0000-0003-4767-2034}}%
	\thanks{Department of Mathematics, University of Munich, Munich, Germany;
	DZ BANK AG Deutsche Zentral-Genossenschaftsbank, Platz der Republik,
	60325 Frankfurt am Main, Germany;
	\url{http://www.christianfries.com}.}
}
\title{Faster Forward Sensitivities}
\subtitle{\small Reduced stochastic hedge ratios from pathwise algorithmic differentiation\\[2ex]
{\normalfont\small \cpfversion}}
\date{May 12, 2026}

\maketitle

\begin{abstract}
Monte-Carlo valuation engines can generate pathwise sensitivities of a derivative value with respect to a high-dimensional vector of model primitives. Hedge ratios with respect to market instruments are then linked to these primitive sensitivities by a pathwise linear relation. Solving this relation independently on every simulated path may be expensive, unstable, and unnecessarily high-dimensional.

This paper studies reduced stochastic hedge ratios of the form $\phi_j^r=\sum_{q=1}^r\xi_j^qX_q$, where the number of solution basis functions is much smaller than the number of Monte-Carlo paths. The hedge-instrument sensitivity tensor is not replaced by its own basis expansion; it is retained through empirical averages over the simulated paths. The basis ansatz alone does not determine the coefficients, so two coefficient criteria are distinguished. The first minimizes the full empirical pathwise residual $\sum_\ell\|A_\ell\phi_\ell^r-b_\ell\|_2^2$. The second is a projected moment equation requiring $\langle A\phi^r-b,Y_s\rangle_N=0$ for selected test functions. The special case $Y_s=X_s$ is the usual Galerkin choice; different test functions give a Petrov--Galerkin formulation. The criteria coincide in special cases but differ when the hedge-instrument sensitivities are path-dependent. The paper gives the tensor and matrix forms of both reductions, discusses regularization and conditioning, and records implementation considerations.

The constructions are motivated by sensitivity-based margin valuation adjustment and replication-consistent liquidity forecasting, where pathwise primitive sensitivities have to be converted into hedge ratios with respect to market instruments.
\end{abstract}

\microtypesetup{protrusion=false}
\tableofcontents
\microtypesetup{protrusion=true}

\clearpage

\section{Introduction}

Monte-Carlo methods for financial derivatives produce present values $V(0)$, cash-flow information, and sensitivities such as $\frac{\dd V(0)}{\dd P_j(0)}$ with respect to observed market instruments that enter the model calibration. In some applications, pathwise forward sensitivities $\frac{\dd V(t,\omega)}{\dd P_j(t,\omega)}$ are relevant. Exemplary applications are sensitivity-based margin valuation adjustment calculations~\cite{Fries2017ForwardSensitivities} and liquidity forecasting~\cite{Fries2025LiquidityForecasting}.

\medskip

A typical simulation produces, at a time $t$, a vector of model primitives
\begin{equation}
	M(t,\omega)=\left(M_1(t,\omega),\ldots,M_n(t,\omega)\right)\in\Rset^n,
\end{equation}
where the components may be forward rates, zero-coupon bond prices, discount factors, volatility factors, credit-state variables, or other quantities used by the valuation model. Let $V(t,\omega)$ denote the on-path value process of a derivative, understood as the pathwise sum of future cash flows under the chosen discounting or numeraire convention.

Algorithmic differentiation, in particular reverse-mode or adjoint algorithmic differentiation, can provide the primitive sensitivity vector
\begin{equation}
	b_i(t,\omega)=\frac{\dd V(t,\omega)}{\dd M_i(t,\omega)},
	\qquad i=1,\ldots,n.
\end{equation}
Hedging, however, is usually performed with market instruments. Let
\begin{equation}
	P(t,\omega)=\left(P_1(t,\omega),\ldots,P_m(t,\omega)\right)\in\Rset^m
\end{equation}
denote the chosen hedge instruments and define the stochastic hedge ratios
\begin{equation}
	\phi_j(t,\omega)=\frac{\dd V(t,\omega)}{\dd P_j(t,\omega)},
	\qquad j=1,\ldots,m.
\end{equation}
If
\begin{equation}
	A_{ij}(t,\omega)=\frac{\dd P_j(t,\omega)}{\dd M_i(t,\omega)},
\end{equation}
then the chain rule gives the pathwise linear relation
\begin{equation}
	A(t,\omega)\phi(t,\omega)=b(t,\omega).
	\label{eq:pathwise-system-intro}
\end{equation}
The direct approach solves \eqref{eq:pathwise-system-intro} separately for every Monte-Carlo path. This creates $Nm$ pathwise hedge unknowns for $N$ paths and may require many small, possibly ill-conditioned, linear or least-squares solves.

The reduction considered here represents only the hedge ratios in a finite empirical solution basis,
\begin{equation}
	\phi_j^r(t,\omega)=\sum_{q=1}^r \xi_j^q X_q(\omega),
	\qquad r\ll N,
\end{equation}
while the pathwise hedge-instrument sensitivities $A(t,\omega)$ are kept in the empirical equations. This removes the path dimension from the unknowns: the coefficients are $mr$ numbers, not $Nm$ pathwise values.

The basis ansatz does not by itself specify how the coefficients should be chosen. Two natural choices are used in this paper.
\begin{enumerate}[(i)]
	\item The empirical $L^2$ formulation minimizes the original pathwise residual $A\phi^r-b$ over all hedge ratios in the basis.
	\item The projected formulation matches residual moments against selected test functions. Using the same functions for trial and test gives the Galerkin case; using different functions gives a Petrov--Galerkin system.
\end{enumerate}
The first formulation gives the best empirical fit to the original sensitivity equation in the selected primitive-space metric. The second formulation matches low-dimensional residual moments and avoids forming products of the type $A^{\transposed}A$. They agree in special cases, for example when the residual can be made zero, but generally differ for stochastic $A$.

The rest of the paper is organised as follows. Section~\ref{sec:setup} defines the pathwise sensitivity equation. Section~\ref{sec:basis} introduces the empirical bases and the hedge-ratio ansatz. Section~\ref{sec:formulations} derives the empirical $L^2$ and projected coefficient systems. Section~\ref{sec:regularization} discusses solution and regularization. Section~\ref{sec:discussion} discusses special cases, approximation, stability, and basis choice. Section~\ref{sec:algorithm} gives implementation and complexity considerations.
Appendix~\ref{sec:nonorthogonal} gives formulas for non-orthonormal bases.

\subsection{Relation to existing methods and applications}

The method is close in spirit to regression-based Monte-Carlo methods and Galerkin projection. A stochastic quantity is represented in a low-dimensional basis and coefficients are estimated from simulated paths. In least-squares Monte-Carlo and American Monte-Carlo methods, the projected quantity is typically a conditional expectation or continuation value; see, for example, \cite{LongstaffSchwartz2001,FriesLectureNotes2007}. Here the target is not a conditional expectation itself, but a sensitivity equation linking model primitives to hedge instruments. Allowing different solution and test bases gives the usual Petrov--Galerkin generalization.

The calculation of sensitivities in Monte-Carlo simulations has a large literature, including finite-difference, pathwise, likelihood-ratio, and algorithmic-differentiation methods; see \cite{BroadieGlasserman1996,Glasserman2004}. Adjoint algorithmic differentiation provides a particularly efficient way to compute many sensitivities of one scalar valuation output and has become a standard tool for Monte-Carlo Greeks; see \cite{GilesGlasserman2006,CapriottiGiles2011,HomescuAAD2011,GriewankWalther2008}. Stochastic automatic differentiation extends this viewpoint to random variables and stochastic operators such as expectations, conditional expectations, and indicator functions; see \cite{FriesAutoDiff4MonteCarlo}. For products whose valuation contains stochastic operators, for example callable or Bermudan products, the direct use of AAD requires additional care; see \cite{AntonovAADforBermudan2017,CapriottiJiangMarcinaAAD2016,FriesAutoDiff4AmericanMonteCarlo}. Discontinuous payoffs and indicator functions are treated in \cite{Fries2018DiscontinuousIndicators}.

The forward-sensitivity construction in \cite{Fries2017ForwardSensitivities}, building on stochastic automatic differentiation \cite{FriesAutoDiff4MonteCarlo,FriesAutoDiff4AmericanMonteCarlo}, shows that stochastic forward sensitivities can be obtained by a single stochastic backward automatic-differentiation sweep followed by a conditional-expectation step. That work is complementary to the present paper. It explains how to obtain the pathwise primitive sensitivities $b_{\ell i}=\dd V/\dd M_i$ efficiently. The present paper starts from these primitive sensitivities, together with the hedge-instrument sensitivities $A_{\ell i j}=\dd P_j/\dd M_i$, and studies reduced systems for converting them into stochastic hedge ratios with respect to market instruments.

The role of basis functions is related but different in the two settings. In \cite{Fries2017ForwardSensitivities}, basis functions are used to approximate conditional expectations of stochastic derivatives. In the present paper, basis functions determine the admissible stochastic hedge rule and, in the projected formulation, the residual moments to be enforced. The warning from the forward-sensitivity literature remains relevant: a poor basis may project away economically important state information. For example, callable products may require basis functions or test functions that encode exercise-region information.

The approach also differs from first solving the pathwise systems $A_\ell\phi_\ell=b_\ell$ and then regressing the resulting hedge ratios. The reduced equations avoid this intermediate pathwise solve and can remain well-defined when the pathwise systems are singular, non-unique, or noisy. Algorithmic differentiation supplies the raw pathwise derivatives $b_{\ell i}$ and $A_{\ell i j}$; the reduced coefficient problem turns these derivatives into a parsimonious stochastic hedge rule.

There is a related model-to-market-sensitivity literature. Li~\cite{Li1999ModelCalibrationHedging} studies the translation of model-parameter sensitivities into delta and vega equivalents with respect to hedging instruments in calibrated models. CVA sensitivity methods similarly use calibration Jacobians to convert model sensitivities into market sensitivities; see, for example, \cite{CrepeyLiNguyenSaadeddine2024}. Regression sensitivities for standardized risk factors in SIMM and FRTB are considered in \cite{AlbaneseCaenazzoFrankel2016}, where ridge regularization is used to obtain stable sensitivities. These works are close in motivation to the present paper. The distinction is that the construction here is pathwise and stochastic: the hedge ratios are functions of the simulated state, while the pathwise hedge-instrument sensitivity tensor is retained in the reduced empirical equations.

Several approaches also determine hedge or risk coefficients by regression or simulation. Hedged Monte-Carlo methods include a hedge strategy in the Monte-Carlo valuation and estimate hedge coefficients together with the price \cite{PottersBouchaudSestovic2001}. Avellaneda and Gamba \cite{AvellanedaGamba2002} characterize Monte-Carlo hedge ratios with respect to input prices through moments of simulated cash flows. Regression-based methods for sensitivities and hedging, especially for American or Bermudan products, are developed in \cite{BelomestnyMilsteinSchoenmakers2007,WangCaflisch2010,JainLeitaoOosterlee2017}. These methods share the use of simulated paths and finite-dimensional approximations, but they do not form the reduced coefficient system from the pathwise relation $A_\ell\phi_\ell=b_\ell$.

Applications to sensitivity-based initial-margin simulation and margin valuation adjustment provide one motivation for forward sensitivities; see \cite{Fries2017ForwardSensitivities,FriesLandgrafViehmann2018}. Related comparisons of forward and backward differentiation for forward sensitivities are discussed in \cite{FriesBackToTheFuture2018}. Issues caused by discontinuous payoffs and exercise indicators are well known in Monte-Carlo sensitivity estimation \cite{Glasserman2004,PiterbargDeltasOfCallableLIBORExotics2004,Fries2018DiscontinuousIndicators} and motivate careful treatment of basis selection and residual diagnostics.

\section{Pathwise sensitivities and hedge ratios}
\label{sec:setup}

We fix a time $t$ throughout the derivation and suppress $t$ from the notation when there is no ambiguity. Let $\omega_1,\ldots,\omega_N$ denote Monte-Carlo paths. On path $\omega_\ell$ write
\begin{equation}
	M_\ell=M(t,\omega_\ell),
	\qquad
	P_{\ell j}=P_j(t,\omega_\ell),
	\qquad
	V_\ell=V(t,\omega_\ell).
\end{equation}
The primitive sensitivity vector is
\begin{equation}
	b_{\ell i}:=\frac{\dd V(t,\omega_\ell)}{\dd M_i(t,\omega_\ell)},
	\qquad \ell=1,\ldots,N,\quad i=1,\ldots,n.
\end{equation}
The hedge-instrument sensitivity tensor is
\begin{equation}
	A_{\ell i j}:=\frac{\dd P_j(t,\omega_\ell)}{\dd M_i(t,\omega_\ell)},
	\qquad \ell=1,\ldots,N,\quad i=1,\ldots,n,
	\quad j=1,\ldots,m.
\end{equation}
Thus
\begin{equation}
	A=(A_{\ell i j})\in\Rset^{N\times n\times m}.
\end{equation}

The stochastic hedge ratios are the pathwise coefficients
\begin{equation}
	\phi_{\ell j}:=\phi_j(t,\omega_\ell)
	=\frac{\dd V(t,\omega_\ell)}{\dd P_j(t,\omega_\ell)},
	\qquad \ell=1,\ldots,N,
	\quad j=1,\ldots,m.
\end{equation}
When the chain rule can be applied pathwise, the primitive sensitivities satisfy
\begin{equation}
	\sum_{j=1}^m A_{\ell i j}\phi_{\ell j}=b_{\ell i},
	\qquad \ell=1,\ldots,N,
	\quad i=1,\ldots,n.
	\label{eq:pathwise-linear-equations}
\end{equation}
Equivalently, if $A_\ell=(A_{\ell i j})_{i,j}\in\Rset^{n\times m}$, $\phi_\ell=(\phi_{\ell j})_j\in\Rset^m$, and $b_\ell=(b_{\ell i})_i\in\Rset^n$, then
\begin{equation}
	A_\ell \phi_\ell=b_\ell,
	\qquad \ell=1,\ldots,N.
	\label{eq:pathwise-matrix-equations}
\end{equation}

\begin{remark}[Exact and approximate hedging]
Equation~\eqref{eq:pathwise-matrix-equations} may fail to have a unique solution on a given path. This may happen because the hedge instruments do not span all primitive sensitivities, because there are more hedge instruments than primitive directions, or because $A_\ell$ is ill-conditioned. The reduced construction below is therefore stated through coefficient problems that can be solved directly, in least-squares form, or with regularization. Exact pathwise representability is not required for the algebraic construction.
\end{remark}

\section{Basis representation of stochastic hedge ratios}
\label{sec:basis}

Let $X_1,\ldots,X_r$ be scalar solution basis functions on the simulated path space. These functions determine the admissible hedge-ratio representation and may depend on the state at time $t$, selected path features, or market observables available at time $t$. Define
\begin{equation}
	X_{\ell k}:=X_k(\omega_\ell),
	\qquad \ell=1,\ldots,N,
	\quad k=1,\ldots,r.
\end{equation}
The empirical inner product is
\begin{equation}
	\la U,V\ra_N:=\frac{1}{N}\sum_{\ell=1}^N U(\omega_\ell)V(\omega_\ell).
\end{equation}
For the main text we assume empirical orthonormality,
\begin{equation}
	\la X_k,X_q\ra_N
	=\frac{1}{N}\sum_{\ell=1}^N X_{\ell k}X_{\ell q}
	=\delta_{kq}.
	\label{eq:empirical-orthonormality}
\end{equation}
The non-orthonormal case is given in Appendix~\ref{sec:nonorthogonal}.

Let
\begin{equation}
	\mathcal U_r:=\spanop\{X_1,\ldots,X_r\}.
\end{equation}
For a scalar pathwise random variable $U$, the empirical orthogonal projection onto $\mathcal U_r$ is
\begin{equation}
	\Piop_rU=\sum_{s=1}^r\la U,X_s\ra_NX_s.
\end{equation}
We approximate each stochastic hedge ratio by an element of $\mathcal U_r$:
\begin{equation}
	\phi_j^r(t,\omega)=\sum_{q=1}^r \xi_j^q X_q(\omega),
	\qquad j=1,\ldots,m.
	\label{eq:basis-ansatz-continuous}
\end{equation}
On the Monte-Carlo paths,
\begin{equation}
	\phi_{\ell j}^r=\sum_{q=1}^r \xi_j^q X_{\ell q},
	\qquad \ell=1,\ldots,N,
	\quad j=1,\ldots,m.
	\label{eq:basis-ansatz-paths}
\end{equation}
The residual generated by this ansatz is
\begin{equation}
	R_{\ell i}(\xi):=
	\sum_{j=1}^m A_{\ell i j}\sum_{q=1}^r\xi_j^qX_{\ell q}-b_{\ell i},
	\qquad i=1,\ldots,n.
	\label{eq:residual-definition}
\end{equation}
Equivalently, with $x^q=(\xi_j^q)_{j=1}^m\in\Rset^m$,
\begin{equation}
	\phi_\ell^r=\sum_{q=1}^r x^q X_{\ell q},
	\qquad
	R_\ell(\xi)=A_\ell\phi_\ell^r-b_\ell.
\end{equation}
The remaining question is how to choose the coefficients $\xi_j^q$. The next section gives two choices.

For the projected formulation we also allow a separate test basis $Y_1,\ldots,Y_p$ with path values
\begin{equation}
	Y_{\ell s}:=Y_s(\omega_\ell),
	\qquad \ell=1,\ldots,N,
	\quad s=1,\ldots,p.
\end{equation}
Using $Y_s=X_s$ and $p=r$ gives the standard Galerkin case. Allowing $Y$ to differ from $X$ gives a Petrov--Galerkin projected moment system. This separates the representation of the hedge rule from the residual moments one wants to enforce.

\section{Reduced coefficient formulations}
\label{sec:formulations}

This section derives two coefficient systems for the same hedge-ratio ansatz \eqref{eq:basis-ansatz-continuous}. The empirical $L^2$ formulation minimizes the full pathwise residual. The projected formulation matches residual moments against the basis.

\subsection{Empirical L2 residual minimization}
\label{sec:ls-formulation}

The direct reduced least-squares problem is
\begin{equation}
	\widehat\xi^{\mathrm{LS}}
	=
	\argmin_{\xi}
	\frac{1}{N}\sum_{\ell=1}^N
	\norm{A_\ell\phi_\ell^r-b_\ell}_2^2.
	\label{eq:pathwise-ls-objective}
\end{equation}
In matrix notation this is the problem suggested by the expression $\norm{AXc-b}_2^2$. Define the design matrix
\begin{equation}
	\mathcal D_{(\ell,i),(j,q)}:=A_{\ell i j}X_{\ell q},
	\qquad
	y_{(\ell,i)}:=b_{\ell i},
\end{equation}
with rows indexed by $(\ell,i)$ and columns by $(j,q)$. If $z_{(j,q)}=\xi_j^q$, then
\begin{equation}
	\widehat z^{\mathrm{LS}}
	=
	\argmin_{z\in\Rset^{mr}}
	\frac{1}{N}\norm{\mathcal Dz-y}_2^2.
	\label{eq:design-ls}
\end{equation}
The normal equations are
\begin{equation}
	\mathcal G z=h,
	\qquad
	\mathcal G:=\frac{1}{N}\mathcal D^{\transposed}\mathcal D,
	\qquad
	h:=\frac{1}{N}\mathcal D^{\transposed}y.
	\label{eq:ls-normal-flat}
\end{equation}
In tensor form,
\begin{equation}
	G_{jk}^{qp}:=\frac{1}{N}\sum_{\ell=1}^N\sum_{i=1}^n
	A_{\ell i j}A_{\ell i k}X_{\ell q}X_{\ell p},
	\label{eq:G-definition}
\end{equation}
\begin{equation}
	h_j^q:=\frac{1}{N}\sum_{\ell=1}^N\sum_{i=1}^n
	A_{\ell i j}b_{\ell i}X_{\ell q},
	\label{eq:h-definition}
\end{equation}
and the equations read
\begin{equation}
	\sum_{k=1}^m\sum_{p=1}^rG_{jk}^{qp}\xi_k^p=h_j^q,
	\qquad j=1,\ldots,m,
	\quad q=1,\ldots,r.
	\label{eq:ls-normal-tensor}
\end{equation}
The matrix $\mathcal G$ is symmetric positive semidefinite. If the design matrix has full column rank, the unregularized solution is unique.

More generally one may minimize a weighted residual
\begin{equation}
	\frac{1}{N}\sum_{\ell=1}^N
	\norm{W_\ell(A_\ell\phi_\ell^r-b_\ell)}_2^2,
	\label{eq:weighted-pathwise-ls}
\end{equation}
where $W_\ell$ specifies the metric in primitive-sensitivity space. This is relevant when the primitive components have different units or scales.

\subsection{Projected moment matching}
\label{sec:projection}

The projected formulation tests the residual against the test basis $Y_1,\ldots,Y_p$:
\begin{equation}
	\la R_i(\xi),Y_s\ra_N=0,
	\qquad i=1,\ldots,n,
	\quad s=1,\ldots,p.
	\label{eq:residual-orthogonality}
\end{equation}
For $Y_s=X_s$ this is the Galerkin equation $\Piop_r((A\phi^r)_i)=\Piop_r b_i$. For a different test basis it is a Petrov--Galerkin, or projected moment, equation.

Written on the Monte-Carlo paths, \eqref{eq:residual-orthogonality} is
\begin{equation}
	\frac{1}{N}\sum_{\ell=1}^N
	\left(
	\sum_{j=1}^m A_{\ell i j}\sum_{q=1}^r\xi_j^qX_{\ell q}-b_{\ell i}
	\right)Y_{\ell s}=0.
	\label{eq:projected-paths}
\end{equation}
Define
\begin{equation}
	B_{ij}^{sq}:=\frac{1}{N}\sum_{\ell=1}^N A_{\ell i j}X_{\ell q}Y_{\ell s},
	\qquad
	\beta_i^s:=\frac{1}{N}\sum_{\ell=1}^N b_{\ell i}Y_{\ell s}.
	\label{eq:B-beta-definition}
\end{equation}
Then the projected system is
\begin{equation}
	\sum_{j=1}^m\sum_{q=1}^r B_{ij}^{sq}\xi_j^q=\beta_i^s,
	\qquad i=1,\ldots,n,
	\quad s=1,\ldots,p.
	\label{eq:reduced-tensor-system}
\end{equation}
The tensor $B$ is the action of the pathwise sensitivity tensor on the solution basis, tested against the chosen residual moments. No separate basis expansion of $A$ is used.

For numerical solution, flatten
\begin{equation}
	\mathcal B_{(i,s),(j,q)}:=B_{ij}^{sq},
	\qquad
	z_{(j,q)}:=\xi_j^q,
	\qquad
	g_{(i,s)}:=\beta_i^s.
\end{equation}
With, for example,
\begin{equation}
	\operatorname{row}(i,s):=(s-1)n+i,
	\qquad
	\operatorname{col}(j,q):=(q-1)m+j,
\end{equation}
one obtains
\begin{equation}
	\mathcal Bz=g,
	\qquad
	\mathcal B\in\Rset^{np\times mr},
	\quad z\in\Rset^{mr},
	\quad g\in\Rset^{np}.
	\label{eq:flattened-system}
\end{equation}
\subsection{Difference between the two reductions}
\label{sec:comparison}

Both formulations use the same solution space $\mathcal U_r^m$ for the hedge ratios. They differ in the test space for the residual. The empirical $L^2$ formulation satisfies
\begin{equation}
	\la R(\xi),A(X_qe_j)\ra_{N,n}=0,
\end{equation}
where $A(X_qe_j)$ denotes the vector-valued path process $\ell\mapsto A_\ell e_jX_{\ell q}$, $\la u,v\ra_{N,n}=N^{-1}\sum_\ell u_\ell^{\transposed}v_\ell$, and $e_j$ is the $j$-th unit vector in $\Rset^m$. Hence the residual is orthogonal to $A\mathcal U_r^m$. The projected formulation instead imposes orthogonality to the chosen test space $\mathcal W_p^n$, where $\mathcal W_p=\spanop\{Y_1,\ldots,Y_p\}$.

For $n=m=1$ the distinction is already visible:
\begin{align}
	\text{projected:}\quad
	&\sum_q\left(\frac{1}{N}\sum_\ell A_\ell X_{\ell q}Y_{\ell s}\right)\xi^q
	=\frac{1}{N}\sum_\ell b_\ell Y_{\ell s}, \\
	\text{least squares:}\quad
	&\sum_q\left(\frac{1}{N}\sum_\ell A_\ell^2X_{\ell q}X_{\ell s}\right)\xi^q
	=\frac{1}{N}\sum_\ell A_\ell b_\ell X_{\ell s}.
\end{align}
Thus the least-squares formulation weights errors by $A_\ell^2$, while the projected formulation uses one factor of $A_\ell$.

The empirical $L^2$ formulation is preferable when the objective is to minimize the original pathwise chain-rule residual in a chosen primitive-space metric. It gives a symmetric positive semidefinite normal matrix, but it involves products $A^{\transposed}A$, may square condition numbers if solved through normal equations, and depends on the scaling of the primitive components.

The projected formulation is a moment-matching equation. It avoids products $A^{\transposed}A$, can be cheaper to assemble when $m$ is large, and, at the level of the exact projected equations, is invariant under deterministic invertible left transformations of the primitive equations. If an inconsistent projected system is solved in an ordinary Euclidean least-squares norm, the chosen norm again matters unless the metric is transformed accordingly. The additional flexibility of the projected formulation is the separate choice of test functions: one may represent hedges in a smooth or implementable basis $X$ while enforcing moments against a different set of economically relevant diagnostics $Y$. It does not, however, minimize the full pathwise residual; it only controls the tested residual moments.

The two formulations coincide when the residual can be made zero in the reduced space. They also coincide, after the appropriate least-squares interpretation, when $A_\ell$ is deterministic and independent of the path. In general they are different, and the choice should be tied to the hedging objective.

\section{Solving and regularization}
\label{sec:regularization}

For the empirical $L^2$ formulation, a basic regularized problem is
\begin{equation}
	\widehat z_\lambda^{\mathrm{LS}}
	=
	\argmin_{z\in\Rset^{mr}}
	\left\{
	\frac{1}{N}\norm{\mathcal Dz-y}_2^2+
	\lambda\norm{z}_2^2
	\right\},
	\qquad \lambda\ge 0.
	\label{eq:tikhonov-ls}
\end{equation}
For the projected formulation, the analogous regularized moment problem is
\begin{equation}
	\widehat z_\lambda^{\mathrm{G}}
	=
	\argmin_{z\in\Rset^{mr}}
	\left\{
	\norm{\mathcal Bz-g}_2^2+
	\lambda\norm{z}_2^2
	\right\}.
	\label{eq:tikhonov-galerkin}
\end{equation}
Here $\lambda=0$ gives the unregularized least-squares problem whenever the corresponding linear system is overdetermined, underdetermined, or inconsistent. For $\lambda>0$, the usual normal equations are available, but QR, singular-value, or iterative least-squares methods are often preferable.

More general regularizers can encode smoothness, penalize high-order basis coefficients, or enforce proximity to a prior hedge rule $z_0$:
\begin{equation}
	\argmin_z
	\left\{
	\norm{W(Cz-d)}_2^2+
	\lambda\norm{L(z-z_0)}_2^2
	\right\}.
	\label{eq:tikhonov-general}
\end{equation}
Here $(C,d)$ is either $(\mathcal D,y)$ for the full residual formulation or $(\mathcal B,g)$ for the projected formulation, $W$ is an optional weight matrix, and $L$ is a regularization operator.

After solving for $\widehat z$, recover the coefficients $\widehat\xi_j^q$ and reconstruct
\begin{equation}
	\widehat\phi_j^r(t,\omega)=\sum_{q=1}^r\widehat\xi_j^qX_q(\omega),
	\qquad j=1,\ldots,m.
	\label{eq:reconstruction-continuous}
\end{equation}
On the simulated paths,
\begin{equation}
	\widehat\phi_{\ell j}^r=\sum_{q=1}^r\widehat\xi_j^qX_{\ell q}.
	\label{eq:reconstruction-paths}
\end{equation}

\section{Discussion}
\label{sec:discussion}

\subsection{Interpretation and special cases}
\label{sec:special-cases}

\subsubsection{Only the hedge ratios are approximated}

Both reduced formulations approximate
\begin{equation}
	\phi_j\approx \phi_j^r\in\mathcal U_r,
\end{equation}
not
\begin{equation}
	A_{ij}\approx A_{ij}^r.
\end{equation}
The simulated values of $A_{\ell i j}$ enter the empirical equations directly. The distinction between the two formulations is not whether $A$ is projected; it is how the coefficients of the hedge-ratio ansatz are chosen.

\subsubsection{Constant basis}

For $r=1$ with $X_1\equiv1$, the hedge ratios are deterministic. With the constant test function $Y_1\equiv1$, the projected formulation gives
\begin{equation}
	\left(\frac{1}{N}\sum_{\ell=1}^N A_\ell\right)x
	=
	\frac{1}{N}\sum_{\ell=1}^N b_\ell.
\end{equation}
The empirical $L^2$ formulation gives the aggregate least-squares hedge
\begin{equation}
	\left(\frac{1}{N}\sum_{\ell=1}^N A_\ell^{\transposed}A_\ell\right)x
	=
	\frac{1}{N}\sum_{\ell=1}^N A_\ell^{\transposed}b_\ell.
\end{equation}
These are generally different unless $A_\ell$ is deterministic, or unless the residual can be made zero.

\subsubsection{Deterministic hedge-instrument sensitivity matrix}

If $A_\ell=A_0$ is independent of the path, the projected blocks satisfy
\begin{equation}
	B^{sq}=A_0H^{sq},
	\qquad
	H^{sq}:=\la X_q,Y_s\ra_N.
\end{equation}
For the Galerkin choice $Y_s=X_s$ with an empirically orthonormal basis, $H^{sq}=\delta_{sq}$ and the projected system decouples into $A_0x^s=b^s$, where
\begin{equation}
	b_i^s:=\la b_i,X_s\ra_N.
\end{equation}
The empirical $L^2$ equations decouple into $A_0^{\transposed}A_0x^s=A_0^{\transposed}b^s$. If $A_0x^s=b^s$ is solvable, the solutions agree. If it is not solvable, the projected system should itself be interpreted in a least-squares sense to recover the same deterministic least-squares hedge.

\subsubsection{Full empirical path basis}

At the opposite extreme, take $r=N$ and, for the projected formulation, $p=N$ with
\begin{equation}
	X_{\ell q}=Y_{\ell q}=\sqrt{N}\,\indicatorfcn_{\{\ell=q\}}.
\end{equation}
Then the reduced problems decouple path by path. The empirical $L^2$ formulation gives the pathwise least-squares solves, while the projected formulation gives the pathwise equations themselves, up to the basis scaling.

\subsubsection{Coupling across basis functions}

For stochastic $A_\ell$, the projected blocks
\begin{equation}
	B^{sq}=\frac{1}{N}\sum_{\ell=1}^N A_\ell X_{\ell q}Y_{\ell s}
\end{equation}
need not vanish for $s\ne q$. Similarly, the empirical $L^2$ normal matrix generally couples all pairs of solution-basis functions through products $A_\ell^{\transposed}A_\ell X_{\ell q}X_{\ell p}$. Block-diagonal approximations may be useful as preconditioners, but they are not equivalent to the full reduced equations unless the off-diagonal blocks are negligible for the intended objective.

\subsection{Approximation, conditioning, and stability}
\label{sec:stability}

The reduced methods have three main sources of error:
\begin{enumerate}[(i)]
	\item \textbf{Basis error.} The true hedge ratios may not lie in $\mathcal U_r^m$.
	\item \textbf{Sampling error.} Empirical averages approximate population inner products.
	\item \textbf{Linear-algebra error.} The reduced matrices may be ill-conditioned or rank deficient.
\end{enumerate}

\begin{proposition}[Exact recovery in the reduced space]
Assume there exist coefficients $\xi_j^q$ such that
\begin{equation}
	A_\ell\phi_\ell^r=b_\ell,
	\qquad \ell=1,\ldots,N.
\end{equation}
Then these coefficients solve both the empirical $L^2$ normal equations and the projected moment equations for any test basis. If the relevant reduced matrix has full column rank, the coefficient vector is unique for the corresponding formulation.
\end{proposition}

\begin{proof}
If $A_\ell\phi_\ell^r=b_\ell$ for every path, then $R_\ell(\xi)=0$ for every path. Hence the full empirical residual is zero and all projected residual moments vanish. The rank condition gives uniqueness.
\end{proof}

When exact recovery is not possible, the two residuals should be monitored separately. A small projected residual $\norm{\mathcal Bz-g}$ does not necessarily imply a small full residual $N^{-1}\sum_\ell\norm{A_\ell\phi_\ell^r-b_\ell}_2^2$. Conversely, the full empirical $L^2$ solution may be sensitive to the metric and scaling of the primitive components. Weighted residuals such as \eqref{eq:weighted-pathwise-ls} make this choice explicit.

Regularization trades bias for stability. Increasing $\lambda$ shrinks the coefficient vector and can reduce variance in ill-conditioned systems. Excessive regularization, however, may suppress economically meaningful state dependence. In applications, $\lambda$ can be selected by out-of-sample pathwise residuals, projected residuals, hedge stability, or a trading-cost objective.

\subsection{Choice of basis}
\label{sec:basis-choice}

The solution basis should reflect the information set and smoothness expected of the hedge. Possible choices include:
\begin{enumerate}[(i)]
	\item constants and low-order polynomials of Markov state variables;
	\item splines or local basis functions in selected risk factors;
	\item principal components or other low-dimensional factors extracted from the model primitives;
	\item product bases combining time, state, and instrument-specific features;
	\item indicator functions for regimes or exercise regions, if discontinuities are expected.
\end{enumerate}

A solution basis with too few functions may underfit state-dependent hedge ratios. A solution basis with too many functions may overfit Monte-Carlo noise and lead to ill-conditioned reduced systems. Empirical orthonormalisation improves conditioning at the basis level, but it does not guarantee that the reduced matrices are well-conditioned, because they also depend on the stochastic hedge-instrument sensitivities.

For projected moments, the test basis can be chosen separately. Taking $Y=X$ is the simplest Galerkin choice. A smaller or different test basis can enforce selected diagnostics, reduce variance, or produce an overdetermined moment system without changing the functional form of the hedge ratios.

The solution basis must also be admissible for the intended use. If the hedge is to be implemented using information available at time $t$, then the solution basis functions should depend only on information available at time $t$. Test functions used only for calibration diagnostics may be broader, but using future path features in either basis changes the interpretation of the fitted hedge rule.

\section{Algorithm and implementation}
\label{sec:algorithm}

The reduction can be implemented by streaming empirical accumulations over paths.
\begin{enumerate}[(1)]
	\item \textbf{Simulate and differentiate.} For each path $\omega_\ell$, compute $b_{\ell i}$ and $A_{\ell i j}$ using the chosen algorithmic-differentiation implementation.
	\item \textbf{Evaluate the bases.} Compute the solution basis $X_q$ for the hedge ratios. For projected moments, also compute the test basis $Y_s$; the default Galerkin choice is $Y_s=X_s$.
	\item \textbf{Choose the coefficient criterion.} For full empirical $L^2$, accumulate $G$ and $h$ from \eqref{eq:G-definition} and \eqref{eq:h-definition}, or apply the design matrix $\mathcal D$ matrix-free. For the projected formulation, accumulate $B$ and $\beta$ from \eqref{eq:B-beta-definition}.
	\item \textbf{Solve and regularize.} Solve the direct, least-squares, or regularized system appropriate to the chosen criterion.
	\item \textbf{Reconstruct hedge ratios.} Use \eqref{eq:reconstruction-continuous} or \eqref{eq:reconstruction-paths}.
\end{enumerate}

\subsection{Computational complexity}

The direct pathwise method solves $N$ systems of size $n\times m$ and has $Nm$ pathwise hedge unknowns. Both reduced formulations have only $mr$ unknown coefficients.

For the projected formulation, naive assembly of $B$ costs
\begin{equation}
	O(Nnmpr)
\end{equation}
and stores $O(nmpr)$ numbers, where $p$ is the number of test functions. For the empirical $L^2$ normal matrix, naive assembly costs
\begin{equation}
	O(Nnm^2r^2)
\end{equation}
and stores $O(m^2r^2)$ numbers. The latter can be reduced in memory by iterative, matrix-free least-squares methods that apply $\mathcal D$ and $\mathcal D^{\transposed}$ without explicitly forming $\mathcal G$.

The projected formulation may therefore be cheaper when $m$ is large. The full empirical $L^2$ formulation is more directly tied to the pathwise residual. Both constructions are embarrassingly parallel across paths during assembly.

\subsection{Out-of-sample use}

If the solution basis functions depend only on observable state variables at time $t$, the fitted coefficients can be used on new paths or in a live hedge calculation by evaluating $X_q$ at the new state and applying \eqref{eq:reconstruction-continuous}. Test functions $Y_s$ are needed for fitting and diagnostics, but not for reconstructing the hedge ratios. If empirical orthonormalization transformed an original solution basis into $X$, the same transformation must be stored and reused out of sample.

\subsection{Reference implementation}

A reference implementation is available in the development version of finmath-lib \cite{finmathLib}. It is located in the package
\begin{center}
\texttt{net.finmath.montecarlo.automaticdifferentiation.}\\
\texttt{forwardsensitivities}
\end{center}
and exposed via the static function \texttt{ForwardSensitivities.getHedgeRatios}. The method computes projected or least-squares reduced hedge ratios. It takes the specification of the financial derivative, hedge portfolio, and model parameters, the solution and test basis functions, and the reduction method.

The derivative and hedge portfolio are passed as differentiable random variables, which allows their differentials to be queried by object id.

\section{Conclusion}

We have derived reduced methods for obtaining stochastic hedge ratios from pathwise model-primitive sensitivities. The common dimension reduction is the representation of hedge ratios in a finite empirical basis; the hedge-instrument sensitivity tensor is not separately compressed.

Two coefficient criteria are natural. The empirical $L^2$ formulation minimizes the original pathwise residual in a chosen primitive-space metric. The projected formulation matches residual moments against chosen test functions and avoids products $A^{\transposed}A$. The Galerkin case uses $Y=X$; the Petrov--Galerkin case separates hedge representation from residual testing. These formulations agree in exact-recovery and deterministic-sensitivity limits, but differ in general. Practical implementations should therefore report both pathwise residuals and projected residuals, and should choose the formulation according to the intended hedging objective.

Future work should add numerical experiments, compare the two reductions with regression after pathwise solves, investigate basis selection and regularization, and study dynamic rebalancing performance in realistic market models.

\appendix

\section{Non-orthonormal empirical bases}
\label{sec:nonorthogonal}

The orthonormality assumption \eqref{eq:empirical-orthonormality} is mainly a notational and conditioning convenience. If the solution basis is $Z_1,\ldots,Z_r$, one may use
\begin{equation}
	\phi_j^r=\sum_{q=1}^r\eta_j^qZ_q
\end{equation}
directly. For the empirical $L^2$ formulation, replace $X$ by $Z$ in \eqref{eq:G-definition} and \eqref{eq:h-definition}:
\begin{equation}
	G_{jk}^{qp,Z}:=\frac{1}{N}\sum_{\ell=1}^N\sum_{i=1}^n
	A_{\ell i j}A_{\ell i k}Z_{\ell q}Z_{\ell p},
	\qquad
	h_j^{q,Z}:=\frac{1}{N}\sum_{\ell=1}^N\sum_{i=1}^n
	A_{\ell i j}b_{\ell i}Z_{\ell q}.
\end{equation}

For the projected formulation, let $Y_1,\ldots,Y_p$ be the test basis. The moment equations are
\begin{equation}
	\la (A\phi^r)_i-b_i,Y_s\ra_N=0,
	\qquad i=1,\ldots,n,
	\quad s=1,\ldots,p,
\end{equation}
which gives
\begin{equation}
	\sum_{j=1}^m\sum_{q=1}^r \widetilde B_{ij}^{sq}\eta_j^q
	=\widetilde\beta_i^s,
\end{equation}
with
\begin{equation}
	\widetilde B_{ij}^{sq}
	:=\frac{1}{N}\sum_{\ell=1}^N A_{\ell i j}Z_{\ell q}Y_{\ell s},
	\qquad
	\widetilde\beta_i^s
	:=\frac{1}{N}\sum_{\ell=1}^N b_{\ell i}Y_{\ell s}.
\end{equation}
No Gram matrix is needed to write these moment equations. A Gram matrix appears only if one wants explicit coefficients of an orthogonal projection in a non-orthonormal basis.

\newpage

\section*{Notes}

{\small
\noindent Classification:
JEL-class: C15, C63, G13, G17, G32.
\\
\phantom{Classification:}
MSC 2020-class: 65C05, 91G20, 91G60, 65F20, 65F22.
\\
\phantom{Classification:}
ACM-class: G.1.3, G.1.6, G.3, I.6.6, I.6.8.
\\

\noindent Keywords:
Monte Carlo simulation,
algorithmic differentiation,
adjoint algorithmic differentiation,
stochastic automatic differentiation,
forward sensitivities,
stochastic hedge ratios,
least-squares Monte Carlo,
Galerkin projection,
Petrov--Galerkin projection,
regularization,
initial margin valuation adjustment,
liquidity forecasting.
}

\end{document}